\documentclass[sigconf]{acmart}
\def\BibTeX{{\rm B\kern-.05em{\sc i\kern-.025em b}\kern-.08emT\kern-.1667em\lower.7ex\hbox{E}\kern-.125emX}}
\newtheorem{mydef}{Definition}
\newtheorem{mylemma}{Lemma}
\newtheorem{mythm}{Theorem}
\usepackage{tcolorbox}
\usepackage{multirow}
\usepackage[ruled,vlined,linesnumbered]{algorithm2e}
\copyrightyear{2020}
\acmYear{2020}
\setcopyright{acmlicensed}
\acmConference[xxxxxxxxxx]{Unknown Conference}{Unknown Date}{Unknown time}
\acmBooktitle{Unknown Book Title}
\acmPrice{15.00}
\acmDOI{10.1145/1122445.1122456}
\acmISBN{978-1-4503-9999-9/18/06}

\begin{document}

%
\title{DySky: Dynamic Skyline Queries on Uncertain Graphs}
%

\author{Suman Banerjee}
\authornote{The first author is financially supported by the Institute Post Doctoral Fellowship Grant by Indian Institute of Technology, Gandhinagar (Project No. MIS/IITGN/PD-SCH/201415/006).}
\affiliation{%
  \institution{Indian Institute of Technology}
  \city{Gandhinagar-382355}
  \state{India}
}
\email{suman.b@iitgn.ac.in}

\author{Bithika Pal}
\affiliation{%
  \institution{Indian Institute of Technology}
  \streetaddress{1 Th{\o}rv{\"a}ld Circle}
  \city{Kharagpur-721302}
  \country{India}}
\email{bithikapal@iitkgp.ac.in}


%
\renewcommand{\shortauthors}{Banerjee and Pal}

%
\begin{abstract}
Given a graph, and a set of query vertices (subset of the vertices), the dynamic skyline query problem returns a subset of data vertices (other than query vertices) which are not dominated by other data vertices based on certain distance measure. In this paper, we study the dynamic skyline query problem on uncertain graphs (DySky). The input to this problem is an uncertain graph, a subset of its nodes as query vertices,  and the goal here is to return all the data vertices which are not dominated by others. We employ two distance measures in uncertain graphs, namely, \emph{Majority Distance}, and \emph{Expected Distance}. Our approach is broadly divided into three steps: \emph{Pruning}, \emph{Distance Computation}, and \emph{Skyline Vertex Set Generation}. We implement the proposed methodology with three publicly available datasets and observe that it can find out skyline vertex set without taking much time even for million sized graphs if expected distance is concerned. Particularly, the pruning strategy reduces the computational time significantly.
\end{abstract}

%
%


%

%

%
\maketitle

\section{Introduction}
`Skyline' has emerged as an effective multi\mbox{-}criteria decision making operator and hence an extensively researched topic in data management community for almost two decades \cite{chomicki2013skyline}. Borzsony et al. \cite{borzsony2001skyline} fist introduced this operator. Given a set of data points $D$, the skyline operator in it returns the subset of them that are not dominated by other data points present in the dataset. For any two points $d_1$ and $d_2$, we say that $d_1$ dominates $d_2$, if with respect to each dimension $d_1$ is not worse than $d_2$, however, strictly better in at least one dimension. Without loss of generality, in this study, we assume that lower value means better in all dimensions. This problem has been studied in the context of \emph{graph data} as well \cite{zou2010dynamic, zheng2014efficient}. In real\mbox{-}world scenarios, the relationship among agents are uncertain in nature and this uncertainty is caused due to several reasons noisy measurements, unknown values, explicit manipulations, etc. Hence, this kind of situations are modeled as an uncertain graph, where edges are marked with existence probabilities. In case of social networks, these probabilities signify the influence probability between two users, in case of computer networks these signify the successful packet transfer probability between two systems etc. Now, we report some recent  literature on skyline query processing and analysis of uncertain graphs.   

\subsection{Related Work}
After introduced by Borzsony et al. \cite{borzsony2001skyline}, skyline queries have been studied on different kinds of data, for different purposes, with different system architectures, such as road networks \cite{miao2018efficiently, fu2017continuous, zhu2018authentication}, bi-criteria networks \cite{jiang2016k, ouyang2018towards}, uncertain data \cite{zhou2015adaptive, nguyen2015preference}, incomplete data \cite{lee2016optimizing, miao2016k}, streaming data \cite{de2015multicore, liu2019parallelizing}, spatial data \cite{son2017top}, encrypted data \cite{liu2018secure}, knowledge graphs \cite{keles2019skyline}, wireless sensors networks \cite{wang2016geometry}; for route recommendations \cite{yawalkar2019route,jiang2018personalized}, finding perspective customers \cite{yin2018cost}; resisting outliers \cite{jaudoin2017making}, favorite product queries \cite{zhou2016top}; with map reduce architecture \cite{park2017efficient, park2015processing}, multi-core architectures \cite{de2016continuous}, cloud computing framework \cite{huang2018efficient} and so on. Keeping the topic of this paper in our mind, here we elaborate the skyline query processing on probabilistic and uncertain data. He et al.  \cite{he2016probabilistic} studied the skyline query on uncertain time series data and developed a two step methodology for to answer this probabilistically. Park et al. \cite{park2015processing} studied the skyline query processing on uncertain data and proposed parallel algorithms for computing the same using map reduce framework. Zhou et al. \cite{zhou2015adaptive} studied the skyline query processing over uncertain data in distributed environments. Le et al.  \cite{le2016answering} studied the skyline queries on uncertain data to return the user specific relevant results without enumerating all possible worlds. Recently, there are several studies in this directions \cite{zeng2019m, liu2018parallel, liu2019parallelizing}. However, to the best of our knowledge skyline query has not been studied yet in the context of uncertain graphs.
\par Due to different practical applications, in recent times analysis of uncertain graphs have emerged as an important research topic \cite{khan2018uncertain, khan2015uncertain}. Several problems have been studied such as clustering \cite{halim2017efficient, chen2015towards}, embedding \cite{hu2017embedding}, subgraph search \cite{jin2011discovering, chen2015towards}, structural pattern findings \cite{bonchi2014core} and so an. Ke et al. \cite{ke2019budgeted} studied the `budgeted reliability maximization problem', where the goal is to add small number of edges to increase the reliability between a given pair of nodes. ke et al. \cite{ke2019depth} recently studied the $s-t$ reliability problem which asks with how much probability a target node $t$ is reachable from a source node $s$ in a given uncertain graph. Chen et al. \cite{chen2018efficient} studied the frequent pattern finding in uncertain graphs and for this problem enumeration-evaluation algorithm for this problem. Look into \cite{kassiano2016mining} for survey.
\subsection{Our Contribution}
In this paper, we propose the noble problem ``dynamic skyline query on uncertain graphs". Given an uncertain graph with a subset of vertices as query vartices, the goal of this problem is to obtain the subset of the data vertices that are not dominated by the other data vertices with respect to some distance measure from the query vertices. Particularly, we make the following contributions in this paper:
\begin{itemize}
\item We introduce the noble problem ``\emph{\textbf{Dy}namic \textbf{Sky}line Queries on Uncertain Graph Problem}" (\textbf{DySky}).
\item We propose a solution methodology for this problem, which broadly divided into three steps, namely, pruning, distance computation and skyline vertex set generation.  
\item Proposed methodology has been implemented with three benchmark datasets and results show that the pruning strategy leads to less number of candidate nodes.
\end{itemize}  
\subsection{Organization}
Rest of the paper is organized as follows: Section \ref{Sec:PPD} describes required preliminary definitions and then define the problem formally. The proposed methodology has been described in Section \ref{Sec:PM}. In Section \ref{Sec:EE} the experimental evaluations of the proposed methodology has been described. Section \ref{Sec:CFD} draws conclusions and gives future directions.
\section{Preliminaries and Problem Definition} \label{Sec:PPD}
In this section, we present required preliminary concepts and then define the \emph{dynamic skyline queries on uncertain graph} problem formally. Initially, we start with a few basic definitions.
\begin{mydef}[Uncertain Graph]
We denote an uncertain graph by $\mathcal{G}(\mathcal{V}, \mathcal{E}, \mathcal{W}, \mathcal{P})$, where $\mathcal{V}(\mathcal{G})=\{v_1, v_2, \ldots, v_n\}$ is the set of $n$ vertices, $\mathcal{E}(\mathcal{G}) \subseteq \mathcal{V}(\mathcal{G}) \times \mathcal{V}(\mathcal{G})$ is the set of $m$ edges, $\mathcal{W}$ is the distance function that assigns each edge to a positive real number, i.e., $\mathcal{W}:\mathcal{E}(\mathcal{G}) \longrightarrow \mathbb{R}^{+}$, and $\mathcal{P}$ is the existence function that assigns each edge to a probability value, i.e., $\mathcal{P}:\mathcal{E}(\mathcal{G}) \longrightarrow (0,1]$.
\end{mydef}
In our study, we consider only simple, finite, undirected, and weighted graphs. The number of nodes and edges of the graph $\mathcal{G}$ is denoted by $n$ and $m$, respectively. For an edge $e \in \mathcal{E}(\mathcal{G})$ its weight and existence probability is denoted by $\mathcal{W}(e)$ and $\mathcal{P}(e)$, respectively. In the literature, an uncertain graph is conceptualized and analyzed by the \emph{possible world model}, which we define next.
\begin{mydef}[Possible World Semantics] \label{Def:PWS}
 An uncertain graph $\mathcal{G}(\mathcal{V}, \mathcal{E}, \mathcal{W}, \mathcal{P})$ can be conceptualized as the probability distribution over a set of deterministic graphs, which is called as the possible world of the uncertain graph, and denoted as $\mathcal{L}(\mathcal{G})$. Each $G(V, E, W) \in \mathcal{L}(\mathcal{G})$ is obtained from $\mathcal{G}$ by keeping all its vertices, keeping its edges with existing probability, and if an edge of $\mathcal{G}$ is also there in $G$, then $\mathcal{W}(e)=W(e)$. Now, the probability that the deterministic garph $G$ will be generated can be computed by the Equation \ref{Eq:PWS}.
\begin{equation} \label{Eq:PWS}
\mathcal{P}_{G\sqsubseteq \mathcal{G}}= \underset{e \in E(G)}{\prod} \mathcal{P}(e) \underset{e \in E(\mathcal{G}) \setminus E(G)}{\prod} (1-\mathcal{P}(e))
\end{equation}
\end{mydef}

In any deterministic graph $G$, its two vertices $v_i$ and $v_j$ are said to be reachable if there exist a path from between $v_i$ and $v_j$. However, in case of uncertain graphs, the reachability between any two given vertices can be defined in probabilistic way, which we call \emph{reliability}.
\begin{mydef}[Reliability]
Given an undirected, uncertain graph $\mathcal{G}$, the reliability between its any two vertices $v_i$ and $v_j$ is defined as the probability that the vertices $v_i$ and $v_j$ can be reachable from each other. We denote the reliability between the vertices $v_i$ and $v_j$ by $\mathcal{R}^{\mathcal{G}}_{(v_iv_j)}$ and defined by the following equation:
\begin{equation}
\mathcal{R}^{\mathcal{G}}_{(v_iv_j)}=\underset{G \in \mathcal{L}(\mathcal{G})}{\sum} I^{G}_{(v_iv_j)}\mathcal{P}_{G\sqsubseteq \mathcal{G}}.
\end{equation}
Here, $I^{G}_{(v_iv_j)}$ is the boolean variable whose value is $1$ if $v_i$ and $v_j$ are connected in $G$ and $0$ otherwise.
\end{mydef}
In case of a deterministic weighted graph, distance between any two vertices is defined as the sum of individual edge weights constituting shortest path. However, in case of uncertain graphs distance between any two vertices can be defined in many ways. Here, we quote two of them that we use in our study. 
\begin{mydef}[Majority Distance] \cite{potamias2009nearest}
Given an uncertain graph $\mathcal{G}$ and its two vertices $v_i, v_j \in V(\mathcal{G})$, its majority distance is denoted by $dist_{md}(v_i,v_j)$ and defined as the most probable shortest path distance. Mathematically, it can be given by the following equation.
\begin{equation}
dist_{md}(v_i,v_j)=\underset{d}{argmax} \ p_{v_iv_j}(d)
\end{equation}
where $p_{v_iv_j}$ is the shortest path distribution between the vertices $v_i$ and $v_j$ that gives probability value for every distance $d$.
\begin{equation}
p_{v_iv_j}(d)=\underset{G|d_{G}(v_i,v_j)=d}{\sum} \mathcal{P}_{G\sqsubseteq \mathcal{G}}
\end{equation}
\end{mydef}

\begin{mydef}[Expected Distance] 
Given an uncertain graph $\mathcal{G}$ and its two vertices $v_i, v_j \in V(\mathcal{G})$, let $P^{l}_{(v_iv_j)}$ denotes the set of paths upto length $l$. For each path $p_k \in P^{l}_{(v_iv_j)}$, the path probability is defined as
\begin{equation}
\mathbb{P}(p_k) = \frac{\prod_{e \in p_k} \mathcal{P}(e)}{ \sum_{p_j \in P^{l}_{(v_iv_j)}} \prod_{e \in p_j} \mathcal{P}(e)}
\end{equation}

Expected distance between $v_i$ and $v_j$ is defined as the 
\begin{equation}
dist_{E}(v_i,v_j)= \underset{p_k \in P^{l}_{(v_iv_j)}}{\sum}dist(p_k). \mathbb{P}(p_k)
\end{equation}
\end{mydef}
\par For any $p \in \mathbb{Z}^{+}$, $[p]$ denotes the set $\{1, 2, \dots, p\}$. Given a set of $2$ or more dimensional data points $\mathcal{D}$, the problem of \emph{skyline query computation} asks to find out the data points that are not dominated by any other data points in $\mathcal{D}$, which is formally defined in Definition \ref{Def:5}. 
\begin{mydef}[Skyline Query] \label{Def:5}
Given a set of $p$ dimensional data points $\mathcal{D}=\{d_1, d_2, \ldots, d_{|\mathcal{D}|}\}$, we say that $d_i$ dominates $d_j$, if for all $k \in [p]$, $d_i(k) \leq d_j(k)$ and there exist atleast one $k \in [p]$ such that $d_i(k) < d_j(k)$. Skyline of the dataset $\mathcal{D}$ is the subset of the data points that are not dominated by any of the data points in $\mathcal{D}$.
\end{mydef} 
Since past one decade or so, skyline queries have been studied extensively \cite{zou2010dynamic, khan2012finding} in graphs as well, which we define next.
\begin{mydef}[Skyline Query in Graphs]
Given a graph $G(V,E)$, and a subset of vertices $\mathcal{Q}$ (called query vertices), for any two data vertices (vertices that are not query vertices, i.e., $V(G) \setminus \mathcal{Q}$) $v_i$ and $v_j$, we say $v_i$ dominates $v_j$, if $\forall w \in \mathcal{Q}$, $dist(w,v_i) \leq dist(w,v_j)$ and $\exists x \in \mathcal{S}$, such that $dist(x,v_i) < dist(x,v_j)$. The skyline query asks to return data vertices that are not dominated by other data vertices.
\end{mydef}
Though, the skyline query problem has been studied in the context of probabilistic data \cite{atallah2009computing, le2016answering, zhang2019modeling}, to the best of our knowledge this problem has not been studied in the context of uncertain graphs. In this paper, we introduce the problem of finding the dynamic skyline queries on uncertain graphs (DySky), which is defined next.
\begin{mydef}[\textbf{Dy}namic \textbf{Sky}line Queries on Uncertain Graphs] \label{Def:DySky}
Given an uncertain graph $\mathcal{G}$, and a subset of vertices $\mathcal{Q}$ (called query vertices), the problem of dynamic skyline queries on uncertain graphs asks to find out the subset of the data vertices such that none of them are dominated by the other data vertices.
\end{mydef}
Figure \ref{Fig:Example} shows a toy example of an uncertain graph with its majority distance, expected distance, and shortest path distance (for deterministic version) tables, where the skyline vertices are marked in orange color. It is important to observe as the distance measure changes, the skyline vertex set is also getting changed. This motivates us to study the DySky Problem under two different distance measures.

\begin{figure}
\centering
\begin{tabular}{c}
\includegraphics[scale=0.12]{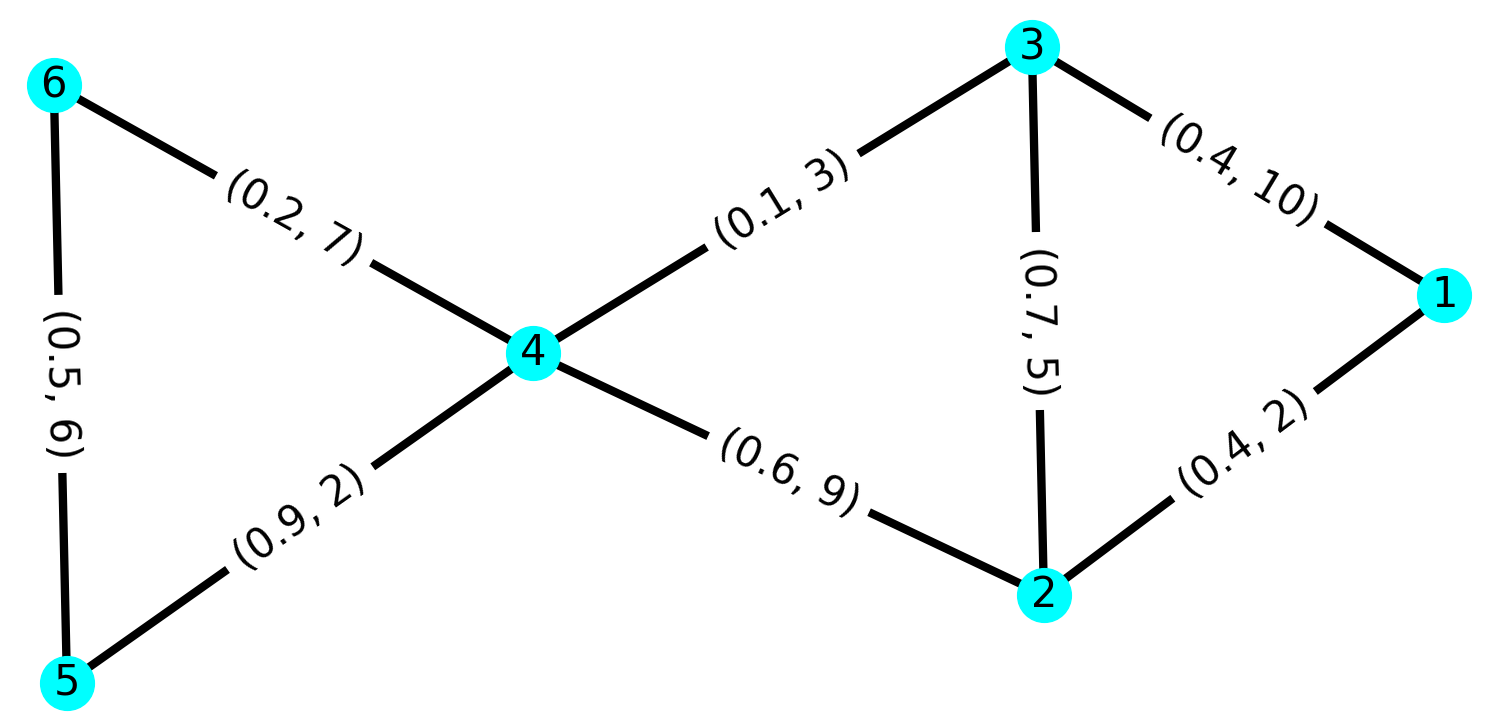}\\
(a) An Uncertain Graph \\
\includegraphics[scale=0.3]{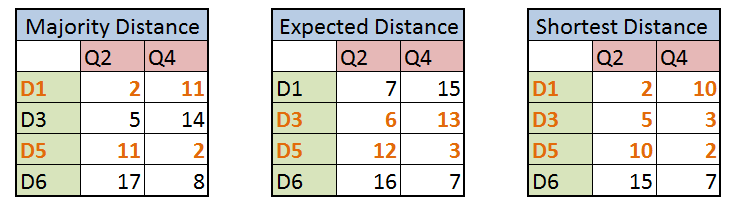}\\
(b) Majority Distance (MD), Expected Distance, \\
and Shortest Path Distance (in deterministic version)
\end{tabular}
\caption{(a) An uncertain graph with $6$ vertices and $8$ edges. The vertices $2$ and $4$ are the Query Vertices (denoted as $Q2$ and $Q4$) and remaining are data vertices (i.e., $D1$, $D3$, $D5$, and $D6$). (b) Different distance tables between the query and data vertices. Skyline vertices in each cases are marked in Orange.}
\label{Fig:Example}
\end{figure}


%


\section{Proposed Methodology} \label{Sec:PM}
Now, we describe the proposed methodology for solving the DySky Problem. Initially, we start by describing an overview of it.
\subsection{Overview}
The proposed methodology is broadly divided into three steps:
\begin{itemize}
\item \textbf{Step 1 (Pruning):} In this step, a subset of the data vertices are returned as the candidate skyline vertices. This step comprises of two subsets. First, pruning is done by performing \emph{Breadth First Search} (B.F.S., henceforth) from the query vertices and subsequently, pruning is done based on distance computation. 
\item \textbf{Step 2 (Distance Computation):} In this step, distance computation is done between the candidate skyline vertices and query vertices. As mentioned previously, in our study we have used majority distance and expected distance.   
\item \textbf{Step 3 (Skyline Vertex Set Generation):} Based on the previously computed distance, any existing skyline finding algorithm can be used to find out the actual skyline vertices. In our study, we have used the \emph{Block Nested Loop (BNL)} Algorithm proposed by Borzsonyi et al. \cite{borzsony2001skyline}.
\end{itemize}
Next, we proceed towards representing the proposed methodology in an algorithmic form and its detailed analysis. 
\subsection{The Algorithm}
Algorithm \ref{Algo:1}, \ref{Algo:2}, and \ref{Algo:3} together constitute the proposed methodology for the DySky Problem. We describe the entire procedure in two subsections. First we start with describing the pruning step.
\subsubsection{The Pruning Step}

Algorithm \ref{Algo:1} describes the B.F.S. and distance based pruning strategies, which takes the uncertain graph, the set of query vertices, and distance threshold as inputs and outputs the candidate skyline vertices. In B.F.S. pruning, from each of the query vertices, B.F.S. trees are constructed to check the connectivity. First, we create the dictionary $\mathcal{D}$. If a query vertex and data vertex is connected and the data vertex has the entry in the dictionary $\mathcal{D}$, the query vertex is included as a value corresponding to this key. Otherwise, a key corresponding to the data vertex is created and the query vertex is added as a value to this `key'. Now, the data vertices that are reachable from all the query vertices are kept as the candidate skyline vertices. Here, the B.F.S. pruning ends.
\par In reality, even if two vertices are connected by a path of large distance (i.e., more than certain threshold), reachability becomes costlier. Hence, to eliminate such vertices, we perform the distance\mbox{-}based pruning. For this purpose, distance between candidate skyline vertex and query vertex is computed. For a candidate skyline vertex, if there exist atleast one query vertex for which the computed distance value is more than the user defined threshold, the candidate skyline vertex set is updated by removing the candidate skyline vertex.

\begin{algorithm}
\SetAlgoLined
\KwData{Uncertain Graph $\mathcal{G}(\mathcal{V}, \mathcal{E}, \mathcal{W}, \mathcal{P})$, The Set of Query Vertices $\mathcal{Q} \subseteq \mathcal{V}(\mathcal{G})$, Distance Threshold $T$}.
\KwResult{Candidate Skyline Vertices $\mathcal{CS} \subseteq \mathcal{V}(\mathcal{G}) \setminus \mathcal{Q}$ }
$\text{Create Dictionary } \mathcal{D}$\;
\For{$\text{All } u \in \mathcal{Q}$}{
   \For{$\text{All } v \in V(\mathcal{G}) \setminus \mathcal{Q}$}{
         \If{$\texttt{Isconnected(uv)}$}{
             \eIf{$v \in \mathcal{D}.Keys()$}{
             $\mathcal{D}[v].values()= \mathcal{D}[v].values() \cup \{u\}$
          }
          {
          $\mathcal{D}.Create\_Key(v)$\;
          $\mathcal{D}[v].Add\_Value(u)$
          }
          }
         }
}
$\mathcal{CS}=\emptyset$\;
\For{$\text{All } u \in \mathcal{D}.Keys()$}{
\If{$\mathcal{D}[u].Values()=\mathcal{Q}$}{
    $\mathcal{CS}=\mathcal{CS} \cup \{u\}$\;
}
}
\For{$\text{All } v \in \mathcal{CS}$}{
   \For{$\text{All } u \in \mathcal{Q}$}{
   \If{$distance(uv) > T$}{
   $\mathcal{CS}=\mathcal{CS} \setminus \{v\}$\;
   }
   
   }
   }
 \caption{Step 1 (B.F.S and Distance based pruning)}
 \label{Algo:1}
\end{algorithm}
\par Any pruning strategy to work correctly should guarantee that it does not remove any skyline vertex. Hence, we show that the Algorithm \ref{Algo:1} is a correct pruning strategy in Lemma \ref{Lemma:1}.

\begin{mylemma} \label{Lemma:1}
The proposed pruning strategy (Algorithm \ref{Algo:1}) is correct.
\end{mylemma} 
\begin{proof}
Follows from the description.
\end{proof}
\par Now, we do an analysis for time and space requirement of Algorithm \ref{Algo:1}. Let $q$ be the number of query vertices, i.e., $|\mathcal{Q}|=q$. For creating the B.F.S. trees rooted at the query vertices requires $\mathcal{O}(q(m+n))$ time. The maximum number of values associated with a `key' in the dictionary $\mathcal{D}$ is of $\mathcal{O}(q)$. Execution time from Line No. $3$ to $14$  and $16$ to $20$ requires $\mathcal{O}(q(n-q)^{2})$ and $\mathcal{O}((n-q)q^{2})$. Now, in distance\mbox{-}based pruning, the number of distance computations is $\mathcal{O}(q(n-q))$. Computing shortest path between  two vertices in a weighted graph with positive edge weights requires $\mathcal{O}(m + n \log n)$ time. Hence, time requirement for distance\mbox{-}based pruning requires $\mathcal{O}(q(n-q)(m + n \log n))$ time. Total time requirement for Algorithm \ref{Algo:1} is of $\mathcal{O}(q(m+n)+nq(n-q)+q(n-q)(m + n \log n))=\mathcal{O}(q(n-q)(m+ n \log n))$. Extra space requirement of Algorithm \ref{Algo:1} is to store the dictionary $\mathcal{D}$, which is of $\mathcal{O}(q(n-q))$, to store the candidate skyline vertices, which is of $\mathcal{O}(n-q)$, and to perform the B.F.S., which is of $\mathcal{O}(n)$. Hence, total space requirement of Algorithm \ref{Algo:1} is of $\mathcal{O}(q(n-q))$. Lemma \ref{Lemma:2} describes the formal statement.
\begin{mylemma} \label{Lemma:2}
Time and space requirement of Algorithm \ref{Algo:1} is of $\mathcal{O}(q(n-q)(m+ n \log n))$ and $\mathcal{O}(q(n-q))$, respectively.
\end{mylemma} 

\subsubsection{Distance Computation and Skyline Vertex Set Generation}
Now, we describe Step $2$ and $3$ of our proposed methodology. It is important to observe that depending upon which distance measure is used (i.e., majority distance or expected distance) Step $2$ will be different. Algorithm \ref{Algo:2} and \ref{Algo:3} describes the last two steps for the majority distance and expected distance, respectively.
\begin{algorithm}
\SetAlgoLined
\KwData{Candidate Skyline Vertices $\mathcal{CS}$ }.
\KwResult{The Skyline Vertex Set $\mathcal{S}$. }
$\text{Generate }|\mathcal{R}| \text{ number of samples graphs}$ \; 
$	\text{Store the graph probabilities in } \mathcal{P}_G[1 \dots |\mathcal{R}|]$\;
$\text{Create\_Matrix }\mathcal{M} \in \mathbb{R}^{|\mathcal{CS}| \times |\mathcal{Q}|}$\;
\For{$\text{All } u \in \mathcal{CS}$}{
\For{$\text{All } v \in \mathcal{Q}$}{
	Create dictionary $Temp$\;
\For{$\text{All } r \in \mathcal{R}$}{
 $d =\text{shortest distance betweeen } u \text{ and }v \text{ in }r$\;
$ Temp[d] = Temp[d] +  \mathcal{P}_G[r]$\;
}
$\mathcal{M}[u][v] = \underset{d}{argmax} \ Temp[d]$\;
}

}
$\mathcal{S} = \text{Apply BNL on } \mathcal{M}$\;
return $\mathcal{S}$\;
\caption{Step $2$ and $3$ (Distance Computation and Skyline Vertex Set Generation) for Majority Distance}
 \label{Algo:2}
\end{algorithm}
\par For the majority distance case, first we generate $|\mathcal{R}|$ number of subgraphs as mentioned in Definition \ref{Def:PWS}, and the corresponding generation probabilities are stored in the array $\mathcal{P}_G$. Next, the majority distance is computed between a candidate skyline vertex and a query vertex. Finally, the BNL Algorithm is applied on the distance matrix $\mathcal{M}$ to obtain the skyline vertex set.
\par Now, we analyze Algorithm \ref{Algo:2} for time and space requirement. As mentioned in Definition \ref{Def:PWS}, generation of $|\mathcal{R}|$ number of subgraphs require $\mathcal{O}(m |\mathcal{R}|)$ time. Using Dijkstra's algorithm computing the shortest path between a pair of vertices requires $\mathcal{O}(m + n \log n)$ time. Hence, execution time from Line $5$ to $14$ requires $\mathcal{O}(q(n-q)|\mathcal{R}|(m + n \log n))$. Now, BNL algorithm requires $\mathcal{O}((n-q)^{2})$ time. Extra space consumed by Algorithm \ref{Algo:2} is to store the array $\mathcal{P}_G$, $Temp$, and the matrix $\mathcal{M}$ which requires $\mathcal{O}(|\mathcal{R}|)$, $\mathcal{O}(|\mathcal{R}|)$, and $\mathcal{O}(q(n-q))$ space, respectively. The formal statement is mentioned in Lemma \ref{Lemma:3}.

\begin{mylemma} \label{Lemma:3}
Time and space requirement of Algorithm \ref{Algo:2} is of $\mathcal{O}(q(n-q)|\mathcal{R}|(m+ n \log n) + (n-q)^{2})$ and $\mathcal{O}(q(n-q) + |\mathcal{R}|)$, respectively.
\end{mylemma} 
Lemma \ref{Lemma:2} and \ref{Lemma:3} together imply the statement mentioned in   Theorem \ref{Th:1}.

\begin{mythm} \label{Th:1}
If majority distance is concerned, the proposed methodology returns the skyline vertex set in $\mathcal{O}(q(n-q)|\mathcal{R}|(m+ n \log n) + (n-q)^{2})$ time and $\mathcal{O}(q(n-q) + |\mathcal{R}|)$ space.
\end{mythm}
It is trivial to observe that Algorithm \ref{Algo:3} just implements the expected distance, and hence, without explanation we move to analyze the algorithm. Assume that maximum degree of the input uncertain graph is $d_{max}$. Hence, the maximum number paths upto length $l$ between any pair of vertices is of $\mathcal{O}(d_{max}^{l})$. Hence, running time from Line $3$ to $13$ is of $\mathcal{O}(q(n-q)ld_{max}^{l})$. Hence, total running time of Algorithm \ref{Algo:3} is of $\mathcal{O}(q(n-q)ld_{max}^{l}+ (n-q)^{2})$. Extra space consumed by the Algorithm \ref{Algo:3} is to store the matrix $\mathcal{M}$, array $Path$, $Prob$ and $dist$ which requires $\mathcal{O}(q(n-q)+ld_{max}^{l})$. Hence, Lemma \ref{Lemma:4} holds.
\begin{mylemma} \label{Lemma:4}
The running time and space requirement of Algorithm \ref{Algo:3} is of $\mathcal{O}(q(n-q)ld_{max}^{l}+ (n-q)^{2})$ and $\mathcal{O}(q(n-q)+ld_{max}^{l})$, respectively.
\end{mylemma}
Lemma \ref{Lemma:2} and \ref{Lemma:4} together imply the statement mentioned in  Theorem \ref{Th:2}.
\begin{mythm} \label{Th:2}
If expected distance is concerned, the proposed methodology returns the skyline vertex set in $\mathcal{O}(q(n-q)(m+ n \log n +ld_{max}^{l}) + (n-q)^{2})$ time and $\mathcal{O}(q(n-q) + |\mathcal{R}|)$ space.
\end{mythm}
 \begin{algorithm}
	\SetAlgoLined
	\KwData{Candidate Skyline Vertices $\mathcal{CS}$ }.
	\KwResult{The Skyline Vertex Set $\mathcal{S}$. }
	$\text{Create\_Matrix }\mathcal{M} \in \mathbb{R}^{|\mathcal{CS}| \times |\mathcal{Q}|}$\;
	\For{$\text{All } u \in \mathcal{CS}$}{
		\For{$\text{All } v \in \mathcal{Q}$}{
			$path$ = Compute all paths from $u$ to $q$ upto length $l$\;
			$prob[1 \dots |path|] = 0$ ; $dist[1 \dots |path|] = 0$\;
			\For{$\text{All } t \in path$}{
				\For{$\text{All } e \in \mathcal{E}(t)$}{
					$prob[t] = prob[t] + \mathcal{P}(e)$\;
					$dist[t] = dist[t] + \mathcal{P}(e) * \mathcal{W}(e) $\;
				}
			}
			$\mathcal{M}[u][v] = {\sum dist}/{\sum prob}$\;
		}
	}
	$\mathcal{S} = \text{Apply BNL on } \mathcal{M}$\;
	return $\mathcal{S}$\;
	\caption{Step $2$ and $3$ (Distance Computation and Skyline Vertex Set Generation) for Expected Distance}
	\label{Algo:3}
\end{algorithm}

\section{Experimental Evaluations} \label{Sec:EE}

\begin{figure*}
\centering
\begin{tabular}{ccc}
\includegraphics[scale=0.15]{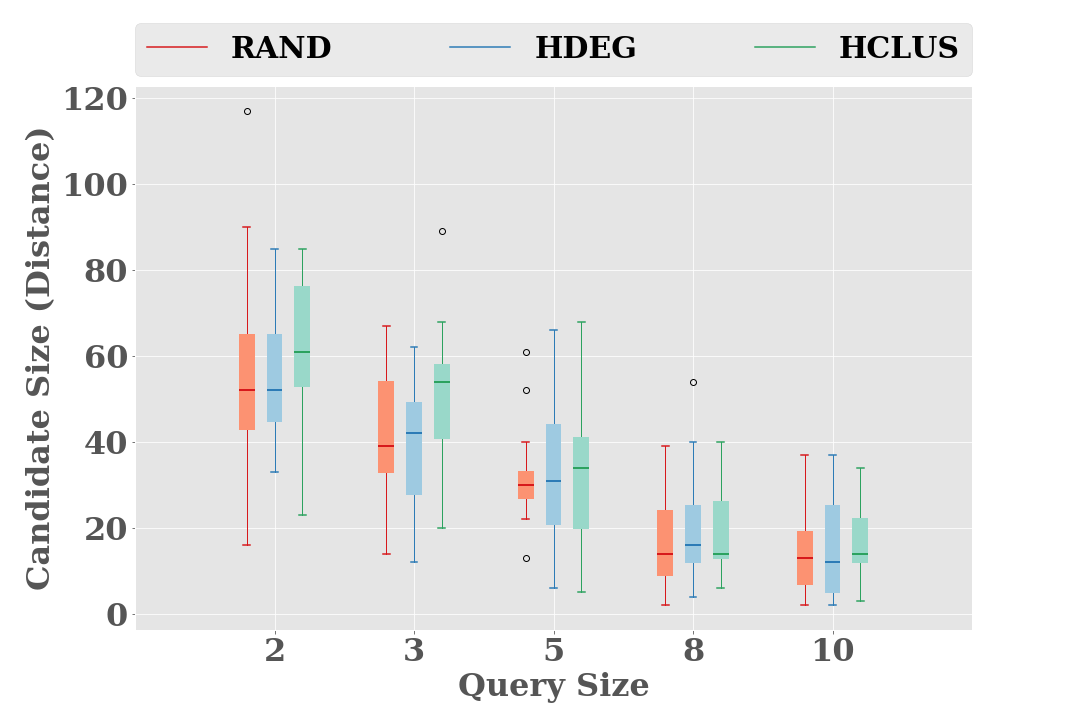} & \includegraphics[scale=0.15]{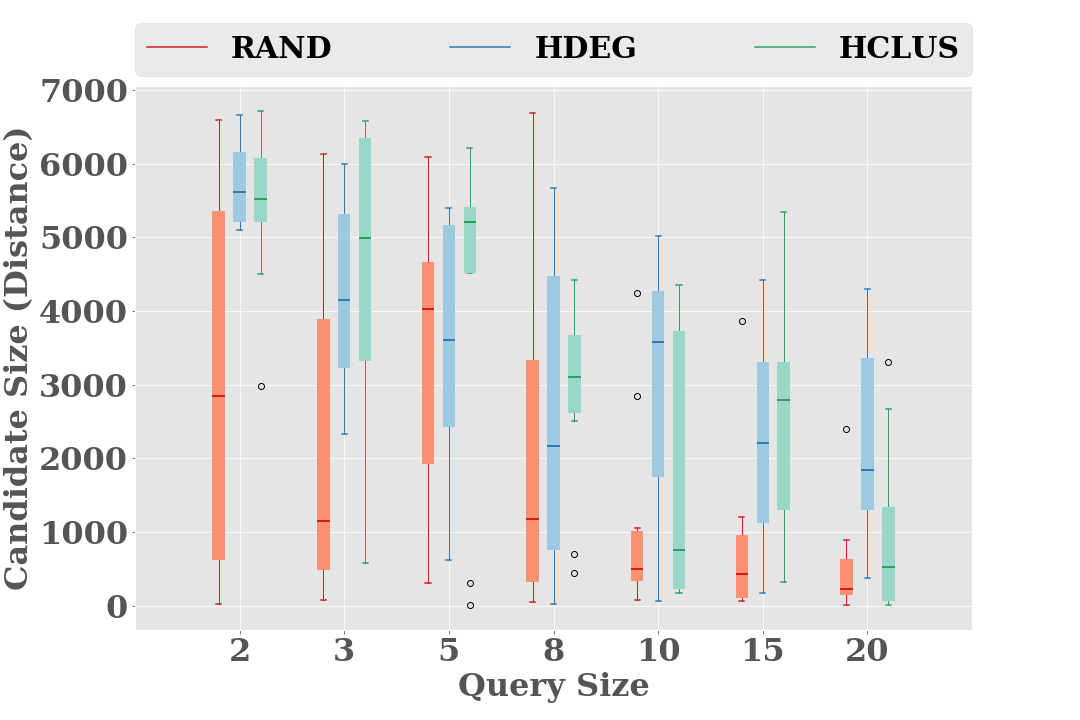} & \includegraphics[scale=0.15]{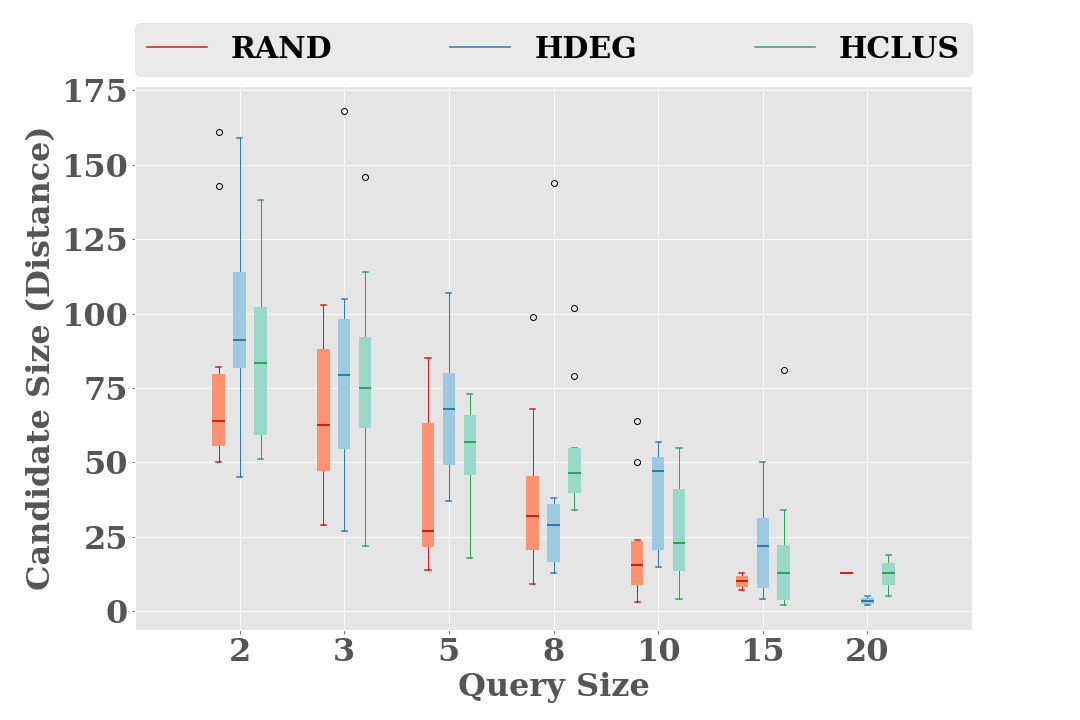}\\
(a) Minnasota Road Network  & (b) P2P Network & (c) USA Road Network \\
\end{tabular}
\caption{Box plot for the candidate skyline size with respect to the query size for the Minnasota Road Network, P2P Network, and USA Road Network datasets.}
\label{Fig:1}
\end{figure*}

\begin{figure*}
\centering
\begin{tabular}{ccc}
\includegraphics[scale=0.15]{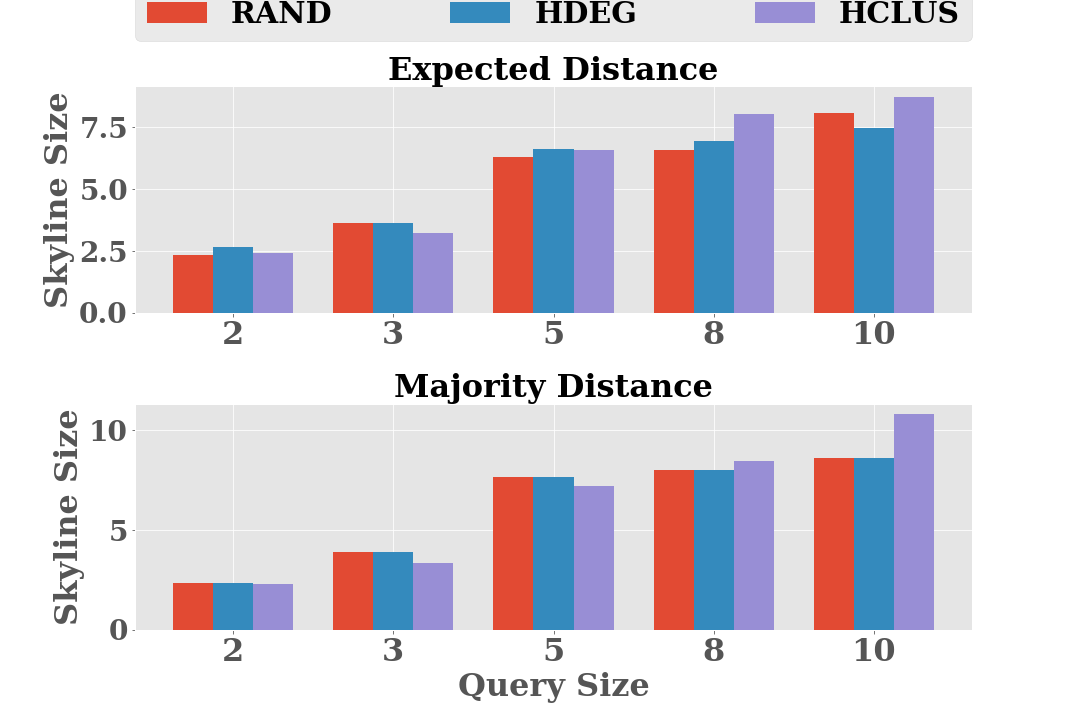} & \includegraphics[scale=0.15]{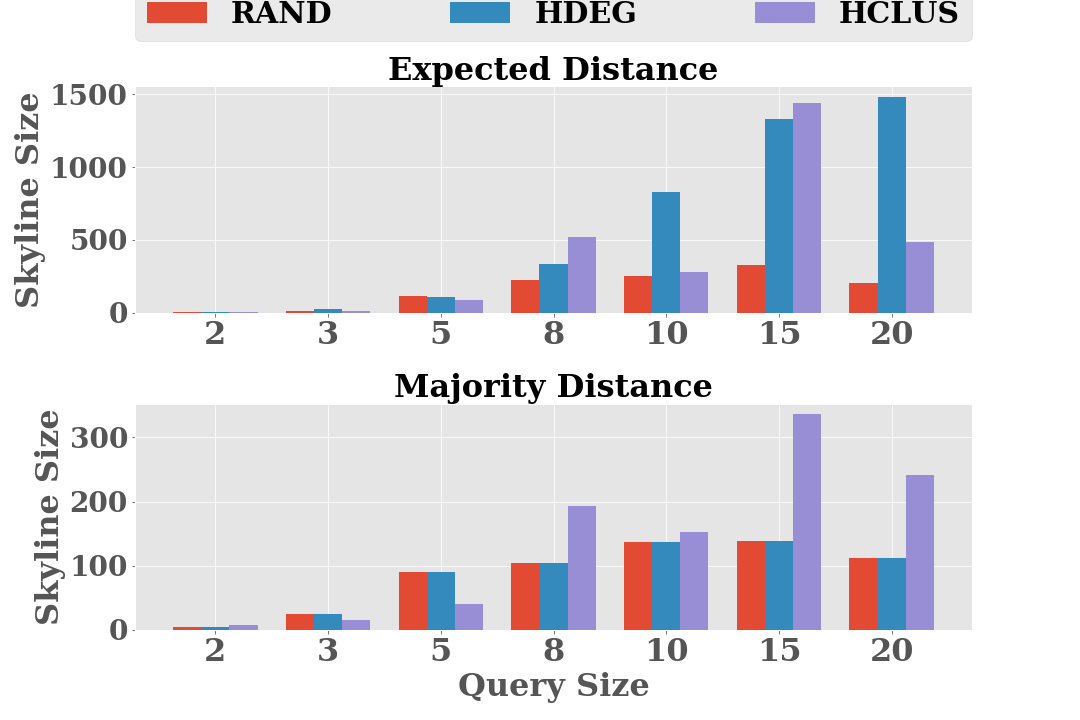} & \includegraphics[scale=0.15]{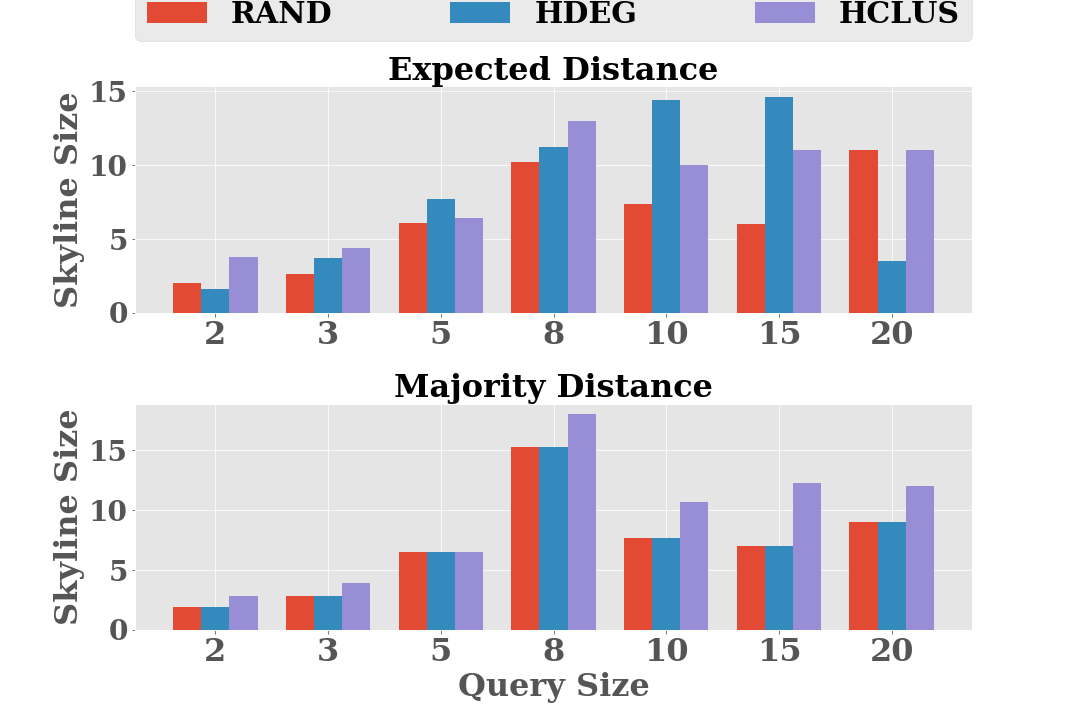}\\
(a) Minnasota Road Network  & (b) P2P Network & (c) USA Road Network \\
\end{tabular}
\caption{Query size Vs. Skyline size plot for the Minnasota Road Network, P2P Network, and USA Road Network datasets.}
\label{Fig:2}
\end{figure*}

In this section we describe the experimental validations of our proposed approach. Initially, we start by describing the datasets.
\subsection{Datasets} In our study, we have used three different datasets appeared in three different contexts described below.
\begin{itemize}
\item \textbf{Minnesota Road Network (MRN)} \cite{nr}: This is a road network dataset  of the Minnasota city. Here, the junctions are represented by the nodes, and if two junctions are connected by a road then the corresponding two vertices are connected by an edge.
\item \textbf{P2P Network} \cite{leskovec2007graph, ripeanu2002mapping}:This dataset contains a sequence of snapshots of the Gnutella peer-to-peer file sharing network from August 2002. There are total of 9 snapshots of Gnutella network collected in August 2002. Nodes represent hosts in the Gnutella network topology and edges represent connections between the Gnutella hosts.
\item \textbf{USA Road Network (URN)} \cite{nr}: This dataset describes a road network from the United States. Here, vertices represent the junctions, and an edge between signifies that the corresponding two junctions are are connected by road.   
\end{itemize}
Please refer to Table \ref{tab:dataset} for basic statistics of the datasets. All the datasets are undirected and unweighted. Probability of existence and weight of each edge is chosen from the intervals $(0,1]$ and $[10,100]$ uniformly at random.
\begin{table}
  \caption{Basic Statistics of the Datasets}
  \label{tab:dataset}
  \begin{tabular}{ccccc}
    \toprule
    Dataset &n&m& Density & Avg. Degree\\
    \midrule
    \textbf{MRN} & 2642 & 3300 & $9.46 \times 10^{-4}$& 2 \\
   \textbf{P2P Network} & 8114 & 26013 & $7.90 \times 10^{-4}$ & 6.41\\
    \textbf{URN} & 129164 & 165435 & $1.98 \times 10^{-5}$ & 2.56\\
  \bottomrule
\end{tabular}
\end{table}
\subsection{Experimental Setup}
In our study the following three different query vertex selection strategies have been adopted:
\begin{itemize}
\item \textbf{RAND}: By this method, to select $k$ query vertices first one is chosen randomly and remaining $(k-1)$ query vertices are chosen from the two hop neighbors of the initially selected vertex uniformly at random.
\item \textbf{HDEG}: By this method, to select $k$ query vertices first the subset of the nodes whose degree is more than a threshold value are marked and a node is chosen uniformly at random as a query vertex. Remaining $(k-1)$ are chosen from the two hop neighbors of the initially selected vertices uniformly at random.
\item \textbf{HCLUS}: This method is exactly the same as HDEG, except the case that, for choosing the first query vertex the subset of vertices are chosen based on the clustering coefficient of nodes.
\end{itemize}
Based on the selection strategy, we choose the query size from the set $\{2, 3, 5, 8, 10, 15, 20\}$. The experiments are repeated for 10 times.
All the algorithms have been implemented with Python 3.5 + NetworkX 2.1 environment on a HPC Cluster with 5 nodes each of them having 64 cores and 160 GB of memory and the implementations are available at \url{https://github.com/BITHIKA1992/Skyline_Uncertain_Graph/}
\subsection{Goals of the Experiments}
The experiments that have been conducted here aim to address the following questions:
\begin{itemize}
\item \emph{Efficiency of the Pruning Strategies}: As the number of query vertices increases, what is the fraction of data vertices removed before distance computation? 
\item \emph{Query Size Vs. Skyline Vertices}: Under different query vertex selection strategies how the cardinality of the skyline vertex set changes with respect to the query size?
\item  \emph{Distance Metric Vs. Skyline Vertices}: For a fixed query selection strategy and query size, how the cardinality of the skyline vertices changes with respect to the distance metric?
\item \emph{Query Selection Strategy Vs. Skyline Vertices}: For a fixed query size and distance metric how the cardinality of the skyline vertices changes with respect to query selection strategy?
\item \emph{Query Size Vs. Computational Time}: For a fixed query size and distance metric, how computational time grows with respect to query size?
\end{itemize}
\subsection{Results and Discussion}
Here, we address the research questions that we have raised.

\subsubsection{Efficiency of the Pruning Strategies} 
As we apply BFS pruning in each dataset, it returns the vertices from the largest component. This reduces $3000$, and $114$ number of vertices for URN and P2P network dataset. For distance based pruning, we have taken the threshold value as $400$, considering 4-hop path with the maximum edge weight $100$. In Figure \ref{Fig:1}, we show the box plot for the candidate size with respect to each query size and the query selection strategy. It can be observed that the candidate size for RAND selection strategy is less than other two, in all the datasets, which is trivial to convince. For P2P network, the inter quartile range is very high compared to other datasets. This is due to the reason of high average degree in the network. Also, for RAND selection strategy, this range is the highest for small query size. This is due to the existence of various small size component in the network. Both the road networks are very sparse and for the large query sizes like 15, 20, the candidate size becomes very small and the variance also reduces. With this sparsity for small road network MRN, it is impossible to find the connected vertices from all the query vertices within the distance of 400. So, we remove the results for query size of 15 and 20 in MRN dataset.

\begin{table*}[h]
 \caption{Computational time requirement (in Secs.) for finding skyline vertex set generation for Minnesota Road Network (MRN), P2P Network, and USA Road Network (URN) Dataset}
    \label{Tab:Time}
\resizebox{1 \textwidth}{!}{ 
\begin{tabular}{ | c | c | c | c | c | c | c | c | c | c |}
    \hline
    \multirow{ 2}{*}{\textbf{Dataset}} & \multirow{ 2}{*}{\textbf{Query Size}} &\multicolumn{3}{|c|}{\textbf{Step 1}} 
    &\multicolumn{2}{|c|}{\textbf{Step 2}}
    & \textbf{Step 3}
       &\multicolumn{2}{|c|}{\textbf{Total Time}}\\
    \cline{3-10} 
     &  & \textbf{Sample Gen} & \textbf{B.F.S Pruning} & \textbf{Distance Pruning} & \textbf{MD Comp. Time} & \textbf{ED Comp. Time} & \textbf{Skyline Comp. Time} & \textbf{MD}& \textbf{ED} \\ \hline
  \multirow{ 5}{*}{MRN}  &2 & 28.6028 & 0.0028  & 0.5797  & 0.2195 & 0.0328 & 0.0019 & 29.4070 & 0.6174 \\ \cline{2-10} 
    & 3 & 29.3554  & 0.0029   & 0.6314  & 0.2232 & 0.0320 & 0.0021 & 30.2152 & 0.6685\\ \cline{2-10} 
    & 5 & 28.9936  & 0.0026   & 0.6565  & 0.2544 & 0.0365 & 0.0025 & 29.9098 & 0.6982 \\ \cline{2-10} 
    & 8 & 30.3126  & 0.0029  & 0.7480  & 0.2443  & 0.0382 & 0.0029  & 31.3109 & 0.7922 \\ \cline{2-10} 
    & 10 & 29.3217  & 0.0029   & 0.8090  & 0.3156 & 0.0487 & 0.0036 & 30.4530 & 0.8643 \\ \cline{2-10}
    \hline
    \hline
      \multirow{ 7}{*}{P2P Network}   &2 & 1585.5587 & 0.0135  & 494.5580  & 16676.0924 & 156.4454 & 0.0026 & 18756.2254 & 651.0197 \\ \cline{2-10} 
    & 3 & 1596.0541  & 0.0142   & 737.8904  & 21824.5948 & 211.6984 & 0.0030 & 24158.5568  &  949.6062 \\ \cline{2-10}
    & 5 & 1594.3928  & 0.0149   & 1122.0655  & 27233.3476 & 299.2759 & 0.0159 & 29949.8369 & 1421.3723 \\ \cline{2-10} 
    & 8 & 1583.5686 & 0.0148  & 1256.8362  & 26993.8274  & 406.0604 & 0.3052  & 29834.5524 & 1663.2167 \\ \cline{2-10}
    & 10 & 1619.4677  & 0.0140   & 1654.2035  & 39048.2783 & 572.3251 & 2.3940 & 42324.3577 & 2228.9368 \\ \cline{2-10}
    & 15 & 1567.8540  & 0.0153   & 2003.2736  & 44549.2078 & 707.3537 & 38.9008 & 48159.2516 & 2749.5436 \\ \cline{2-10}
    & 20 & 1611.6262  & 0.0168  & 2826.1580  & 60793.5038  & 954.3398 & 661.4532  & 65892.7582 & 4441.9679 \\ \cline{2-10}
    \hline
    \hline
      \multirow{ 7}{*}{URN}  &2 & 137954.2686 & 0.1908  & 52.9266  & 0.7025 & 0.0650 & 0.0039 & 138008.0926 & 53.1864 \\ \cline{2-10} 
    & 3 & 139426.9728  & 0.1735  & 48.0920  & 0.7844 & 0.0730 & 0.0041& 139476.0271 & 48.3427 \\ \cline{2-10} 
    & 5 & 138948.9628  & 0.1743   & 54.1615  & 1.3545 & 0.1053 & 0.0063 & 139004.6596 & 54.4476 \\ \cline{2-10}
    & 8 & 163983.4069  & 0.1838  & 51.6949  & 1.2328  & 0.0928 & 0.0043  & 164036.5230 &51.9759 \\ \cline{2-10}
    & 10 & 162627.5420  & 0.1947   & 57.9749  & 1.4499 & 0.1180 & 0.0070 & 162687.1688 &58.2948\\ \cline{2-10}
    & 15 & 115441.9538  & 0.1738   & 67.8932  & 0.8683 & 0.0847 & 0.0167 & 115510.9059 &68.1685 \\ \cline{2-10}
    & 20 & 114801.2682  & 0.1805  & 79.6447  & 0.1371  & 0.0241 & 0.0107  & 114881.2415 & 79.8602\\ \cline{2-10}
    \hline
    \hline
    \end{tabular}
    }
\end{table*}

\subsubsection{Query Size Vs. Skyline Vertices} 
In Figure \ref{Fig:2}, we show the plot for query size Vs. skyline size, with two distance metrics and three query selection strategies. In this part, we describe the comparison of sizes. From all the 10 executions, here we report the mean values for the skyline size. With the increase in query size, the skyline size increases. However, for URN  dataset in Figure \ref{Fig:2}(c), the skyline size decreases for large value of query size. The reason is due to small size of candidate skyline, which can be verified from the Figure \ref{Fig:1}(c). Also, for both the road network datasets the maximum skyline size reaches to approximately 15, whereas for the P2P network it reaches to around 1500. This due to its candidate size. For, both the cases, at large value of query size the ratio of candidate to skyline size is very small. As the number of query vertices increase, the chance of domination decreases.

\subsubsection{Distance Metric Vs. Skyline Vertices} 
In this part, referring to Figure \ref{Fig:2}, we describe the behavior of skyline size with respect to different distance metrics. For the road networks in Figure \ref{Fig:2}(a) and (c), the skyline size is similar in both the datasets. However, for the P2P network in Figure \ref{Fig:2}(b), the skyline size in the expected distance ($\approx$ max 1500) is much more than the majority distance ($\approx$ max 300). The reason lies on the networks high average degree value and the density. As the number of paths increases between a query vertex to a data vertex, the expected distance value is unable to dominate other data vertices. This results into large size of skyline vertex set. This can be verified from Figure \ref{Fig:2}(b), by looking into HDEG and and HCLUS selection strategies, where it differs from the expected distance results. However, for RAND, the size is similar in both the distances. From the experiments, we also observe that for a particular query vertex set the skyline vertices may not be the same from both the distance metrics.

\subsubsection{Query Selection Strategy Vs. Skyline Vertices} 
In this part, referring to Figure \ref{Fig:2}, we describe the behavior of skyline size with respect to different query selection strategies. First, we describe the threshold value selected for HDEG and HCLUS for different datasets. As the P2P network dataset consists of high degree nodes, we select the high degree threshold value as 15, and it returns 440 nodes. In case of both the road networks, the maximum degree is around 5. Hence, for MRN and URN datasets, this threshold value is considered as 2 and 3, respectively. The clustering coefficient threshold is taken as 0 as the clustering coefficient for all the networks are very less. From Figure \ref{Fig:2}, the main observation is that for all the selection strategies the skyline size does not vary much for smaller query size. Whereas, for the large value of query size, HCLUS gives maximum skyline vertices.

\subsubsection{Computational Time} 
 Table \ref{Tab:Time} contains the stepwise computational time requirement to find  skyline vertices for different datasets. From the table, it can be observed that for all the datasets as the query size increases, time requirement for finding out the skyline vertex set also increases. Due to the change in the query size, required time for distance\mbox{-}based pruning, distance and skyline computation (using BNL) increases. Also, for all the datasets, the main time requirement is due to the sample graph generation. As in case of expected distance sample generation is not required, hence, in this distance setting time requirement is much less compared to the majority distance. In particular, for query size $2$, the ratio between the computational time requirement for majority distance to expected distance for MRN, P2P, and URN are $47$, $28$, and $2556$, respectively.
 \par Now, we proceed for the dataset specific observations. For the P2P Network dataset, when the query size increases beyond $10$, there is a sharp increase in the skyline computation time. This is due to the following two reasons. From the Figure \ref{Fig:1}(b) and \ref{Fig:2}(b), it can be observed that candidate and skyline size are more compared to the previous query sizes.

\section{Conclusion and Future Directions} \label{Sec:CFD}
In this paper, we introduce the problem of dynamic skyline queries on uncertain graphs for two different distance measures, namely, majority distance and expected distance. For this problem, we have proposed a methodology having three main steps: pruning, distance computation, and skyline vertex set generation. The proposed methodology has been analyzed to understand its time and space requirement. The experimental results demonstrate that it can find out the skyline vertex set with reasonable computation time, particularly for the expected distance.

 Now, this study can be extended in several directions. It will be an interesting future study to come up with efficient methodology, which can reduce the computational time even for majority distance. One possible way could be parallelizing the sample graph generation. It will be an important future work to provide a sample bound for the majority distance case. The minimum number of samples from the possible world, one should choose to answer the skyline with more than certain threshold probability.




%
\begin{acks}
Authors want to thank Ministry of Human Resource and Development (MHRD), Government of India, for sponsoring the project, E-business Center of Excellence under the scheme of Center for Training and Research in Frontier Areas of Science and Technology (FAST), Grant No. F.No.5-5/2014-TS.VII.
\end{acks}

%
\bibliographystyle{ACM-Reference-Format}
\bibliography{sample}
\end{document}